\documentclass{article}

\usepackage[utf8]{inputenc}
\usepackage{subfigure}

\usepackage[top=3cm,bottom=3cm,left=3cm,right=3cm]{geometry}
\usepackage{amsmath,amsfonts,amssymb,amsthm}
\usepackage{enumerate}
\usepackage{url}
\usepackage{array}
\usepackage[warn]{textcomp}
\usepackage{tikz}
\usetikzlibrary{decorations.pathreplacing,calc}
\usetikzlibrary{decorations.markings}
\usetikzlibrary{decorations.pathmorphing}

\usepackage{todonotes}

\newcommand{\N}{\mathbb{N}}
\newcommand{\Z}{\mathbb{Z}}

\newcommand{\leaf}[1]{|#1|_{\mbox{\scriptsize\textleaf}}}
\newcommand{\size}[1]{|#1|}

\newcommand{\dist}{\mathrm{dist}}
\newcommand{\bigo}{\mathcal{O}}
\newcommand{\IndepSet}{$\textsc{IndependentSet}$} 		
\newcommand{\ISil}{$\mathrm{LIS}$}
\newcommand{\MLIS}{$\mathrm{FLIS}$}

\newcommand{\Hex}{\mathrm{Hex}}
\newcommand{\Squ}{\mathrm{Squ}}
\newcommand{\Tri}{\mathrm{Tri}}
\newcommand{\Cub}{\mathrm{Cub}}
\newcommand{\La}{\lambda}
\newcommand{\green}{\emph{green}}
\newcommand{\yellow}{\emph{yellow}}
\newcommand{\red}{\emph{red}}
\newcommand{\blue}{\emph{blue}}

\newtheorem{theorem}{Theorem}[section]

\newtheorem{lemma}{Lemma}[section]
\newtheorem{proposition}{Proposition}[section]
\newtheorem{corollary}{Corollary}[section]
\theoremstyle{definition}
\newtheorem{definition}{Definition}[section]
\newtheorem{example}{Example}[section]

\newtheorem{remark}{Remark}[section]
\newtheorem{problem}{Problem}[section]

\usepackage{algorithm}
\usepackage{algpseudocode}

\title{Fully leafed induced subtrees\thanks{A.~Blondin Mass\'e is supported by a grant from the National Sciences and Engineering Research Council of Canada (NSERC) through Individual Discovery Grant RGPIN-417269-2013. M.~Lapointe   is supported by a scholarship  CGSD3-488894-2016 from the NSERC and \'E.~Nadeau is supported by a scholarship from the NSERC.}
}
\author{
	Alexandre Blondin Mass\'e 
	\and Julien de Carufel
	\and Alain Goupil
	\and M\'elodie Lapointe
	\and \'Emile Nadeau
	\and \'Elise Vandomme
}
\begin{document}
\maketitle

\begin{abstract}
Let $G$ be a simple graph on $n$ vertices.
We consider the problem \ISil{} of deciding whether there exists an induced subtree with exactly $i \leq n$ vertices and $\ell$ leaves in $G$.
We study the associated optimization problem, that consists in computing the maximal number of leaves, denoted by $L_G(i)$, realized by an induced subtree with $i$ vertices, for $0 \le i \le n$.
We begin by proving that the \ISil{} problem is NP-complete in general and then we compute the values of the map $L_G$ for some classical families of graphs and in particular for the $d$-dimensional hypercubic graphs $Q_d$, for $2 \leq d \leq 6$.
We also describe a nontrivial branch and bound algorithm that computes the function $L_G$ for any simple graph $G$.
In the special case where $G$ is a tree of maximum degree $\Delta$, we provide a $\bigo(n^3\Delta)$ time and $\bigo(n^2)$ space algorithm to compute the function  $L_G$.
\end{abstract}


\section{Introduction}\label{sec:intro}

In the past decades, subtrees of graphs, as well as their number of leaves, have been the subject of investigation from various communities.
For instance in 1984, Payan {\it et al.}~\cite{payan} discussed the maximum number of leaves, called the \emph{leaf number}, that can be realized by a spanning tree of a given graph.
This problem, called the \emph{maximum leaf spanning tree problem} ($\mathrm{MLST}$), is known to be NP-complete even in the case of regular graphs of degree $4$~\cite{garey} and has attracted interest in the telecommunication network community \cite{networks,chen}.
The \emph{frequent subtree mining problem} \cite{deepak}
investigated in the data mining community, has applications in biology.
The detection of subgraph patterns such as induced subtrees is useful in information retrieval~\cite{zaki} and requires efficient algorithms for the enumeration of induced subtrees.
In this perspective, Wasa {\it et al.}~\cite{wasa} proposed an efficient parametrized algorithm for the generation of induced subtrees in a graph.

The center of interest of this paper  are induced subtrees.
The  \emph{induced} property requirement brings an interesting constraint on subtrees, yielding distinctive structures   with respect to other constraints such as in the $\mathrm{MLST}$ problem.
A first result due to Erd\H{o}s \textit{et al.}  in 1986, showed that the problem of finding an induced subtree of a given graph $G$ with more than $i$ vertices is NP-complete~\cite{Erdos-Saks-Sos}.
Another similar famous problem in the error-correcting codes community, called \emph{snake-in-the-box} \cite{Kautz-1958} problem, asks for the length of the longest induced path subgraph in hypercubes and is still open as of today.
Similarly self-avoiding walks, or paths, have been investigated in various lattices \cite{bousquet,duminil}.
When one adds the constraint of being induced, these walks become \emph{thick walks}.
A particular family of thick walks on the square lattice was successfully investigated in \cite{goupil}.

Among induced subtrees of simple graphs, we focus in particular on those with a maximal number of leaves. We call these objects \emph{fully leafed induced subtrees} (FLIS).
Particular instances of the FLIS have recently appeared  in the paper of Blondin Mass\'e \textit{et al.}~\cite{Blondin-FPSAC}, where the authors considered the maximal number of leaves that can be realized by tree-like polyominoes,  respectively polycubes, which are edge connected, respectively face connected,  sets of unit square, respectively cubes.
The investigation of fully leafed tree-like polyforms led to the discovery of a new 3D tree-like polycube
 structure that realizes the maximal number of leaves constraint.
The observation that tree-like polyominoes and polycubes are induced subgraphs of the  lattices $\Z^2$ and $\Z^3$ respectively leads naturally to the investigation of FLIS in general simple graphs, either finite or infinite.

To begin with, we consider the  decision problem, called \emph{leafed induced subtree problem (\ISil)}, and its associated optimization problem, \emph{fully leafed induced subtree problem (\MLIS)}:

\begin{problem}[\ISil]
  Given a simple graph $G$ and two positive integers $i$ and $\ell$, does there exist an induced subtree of $G$ with $i$ vertices and $\ell$ leaves?
\end{problem}

\begin{problem}[\MLIS]
Given a simple graph $G$ on $n$ vertices, what is the maximum number  of leaves, $L_G(i)$,  that can be realized by an induced subtree of $G$ with $i$ vertices, for $i\in\{0,1,\ldots,n\}$?
\end{problem}

If $T$ is an induced subtree of $G$ with $i$ vertices, we say that $T$ is \emph{fully leafed} when its number of leaves is exactly $L_G(i)$. Examples of fully leafed induced subtrees are given in Figure~\ref{F:polyforms}.
\begin{figure}
  \centering
  \includegraphics[scale=0.9]{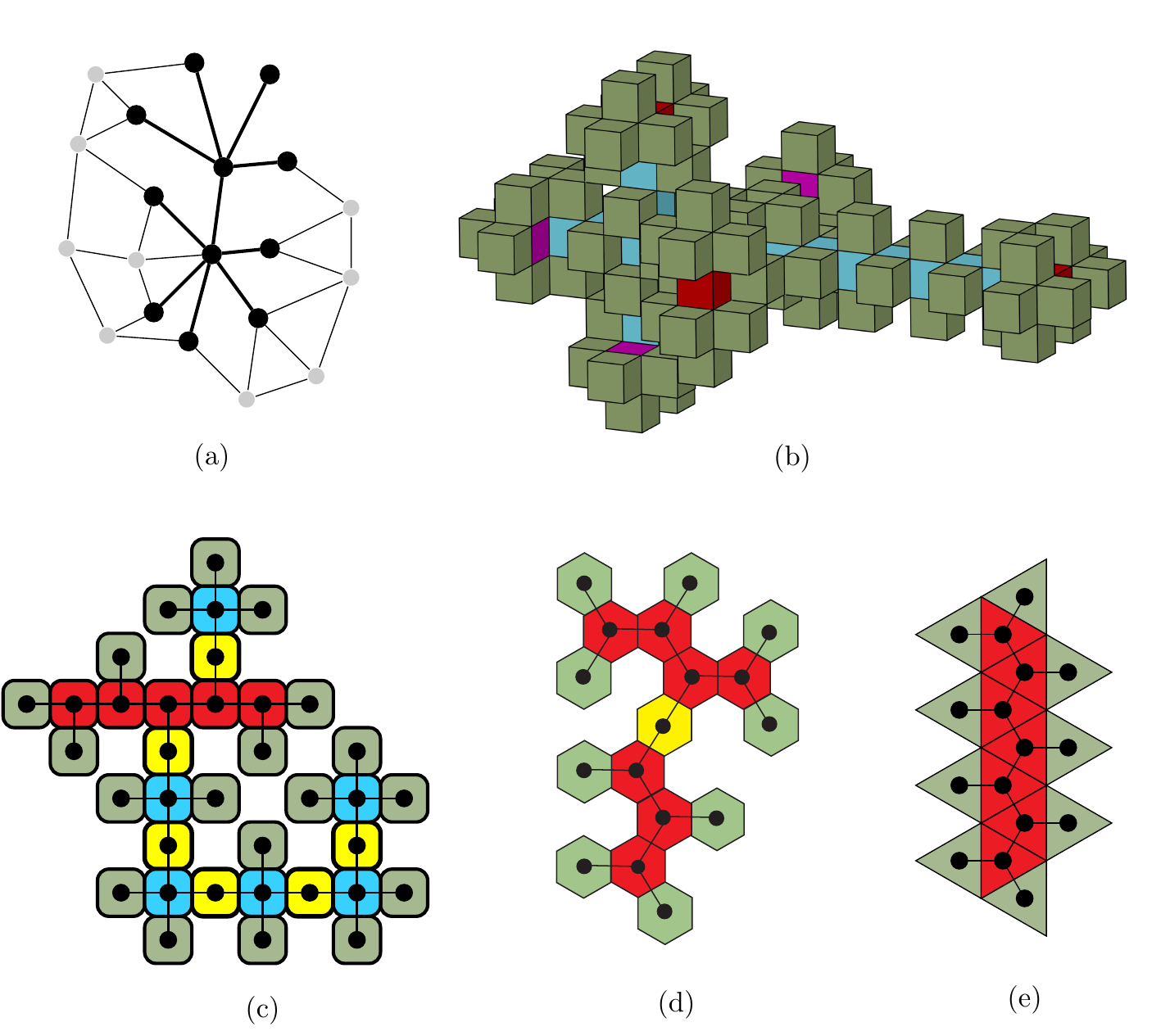}
  \caption{Fully leafed induced subtrees in various graphs. (a) In a finite graph (the subtree of $i = 11$ vertices appears in black). (b) In the cubic lattice. (c) In the square lattice. (d) In the hexagonal lattice. (e) In the triangular lattice. The color of each cell indicates its degree: pink for degree $5$, blue for degree $4$, red for degree $3$, yellow for degree $2$ and green for degree $1$ (the leaves).}\label{F:polyforms}
 \end{figure}
We believe that fully leafed induced subtrees are interesting candidates for the representation of structures appearing in nature and in particular in molecular networks.
Indeed, in chemical graph theory, subtrees are known to be useful in the computation of  a characteristic of chemical graph, called the \emph{Wiener index}, that corresponds to a topological index of a molecule~\cite{Szekely-Wang}.
The results of \cite{Szekely-Wang} and \cite{Blondin-FPSAC} suggest that a thorough investigation of subtrees, and in particular induced subtrees with many leaves, could lead to the discovery of combinatorial structures relevant to chemical graph theory.

This paper establishes fundamental results on fully leafed induced subtrees for further theoretical investigations and their applications.
First, we prove that the problem \ISil{} is NP-complete.
To tackle the problem \MLIS{}, we provide a branch and bound algorithm.
Contrary to a naive algorithm that considers all induced subtrees to compute the maximal number of leaves, the strategy prunes the search space by discarding induced subtrees that cannot be extended to fully leafed subtrees.
When we restrict our investigation to the case of trees, it turns out that the problem \MLIS{} is polynomial.
To achieve this polynomial complexity, our proposed algorithm
 uses a dynamic programming strategy.
Notice that a naive greedy approach cannot work, even in the case of trees, because a fully leafed induced subtree with $n$ vertices is not necessarily a subtree  of  a fully leafed induced subtree with $n+1$ vertices.
All algorithms discussed in this paper are available, with examples, in a public GitHub repository~\cite{GitHub}.

The manuscript is organized as follows.
Basic notions are recalled in Section~\ref{sec:preliminaries} and a proof of the NP-completeness of the decision problem \ISil{} is given. We also study the function $L_G$ in classical families of graphs.
A general branch and bound algorithm to compute \MLIS{} is described  in Section~\ref{sec:algo}.
In Section \ref{sec:tree}, we exhibit a polynomial algorithm to compute the function $L_G$ when $G$ is a tree so that the problem \MLIS{} is in the class P for the particular case of trees.
We conclude the paper in Section~\ref{sec:concl} with some perspectives on future work.


\section{Fully leafed induced subtrees}\label{sec:preliminaries}

We recall some definitions from graph theory and refer the reader to \cite{Diestel--2010} for fundamental notions.
All  graphs considered in this text are simple and undirected unless stated otherwise.
Let $G = (V,E)$ be a graph with vertex set $V$ and edge set $E$.
Given two vertices $u$ and $v$ of $G$, we denote by $\dist(u,v)$ the \emph{distance} between $u$ and $v$, that is the number of edges in a shortest path between $u$ and $v$.
The \emph{degree} of a vertex $u$ is the number of vertices that are at distance 1 from $u$ and is denoted by $\deg(u)$.
We denote by $\size{G}$ the total number $|V|$ of vertices  of $G$ and we  call it the \emph{size} of $G$.
For $U \subseteq V$, the \emph{subgraph of $G$ induced by $U$}, denoted by $G[U]$, is the graph $G[U] = (U, E \cap \mathcal{P}_2(U))$, where $\mathcal{P}_2(U)$ is the set of all subsets of size $2$ of $U$.
Let $T = (V,E)$ be a \emph{tree},  that is to say, a connected and acyclic graph. A vertex $u \in V$ is called a \emph{leaf} of $T$ when $\deg(u) = 1$.
The number of leaves of $T$ is denoted by $\leaf{T}$.
A \emph{subtree of $G$ induced by $U$} is an induced subgraph that is also a tree.

The next definitions and notation are useful in the study of the \ISil{} and \MLIS{} problems.

\begin{definition}[Leaf function]\label{D:leaf-function}
Given a finite or infinite graph $G = (V,E)$, let $\mathcal{T}_G(i)$ be the family of all induced subtrees of $G$ with exactly $i$ vertices. The \emph{leaf function} of $G$, denoted by $L_G$, is the function with domain $\{0,1,2,\ldots,\size{G}\}$ defined by
$$L_G(i)=\max\{\leaf{T} : T\in \mathcal{T}_G(i)\}.$$
As is customary, we set $\max\emptyset = -\infty$.
An induced subtree $T$ of $G$ with $i$ vertices is called \emph{fully leafed} when $\leaf{T} = L_G(i)$.
\end{definition}

\begin{example}
 Consider the graph $G$ depicted in Figure~\ref{fig:leaf_function}. Its leaf function is
 $$
   \begin{array}{c|ccccccccc}
     i &      0 & 1 & 2 & 3 & 4 & 5 & 6 & 7 & 8\\ \hline
     L_G(i) & 0 & 0 & 2 & 2 & 3 & 4 & 4 & 5 & -\infty\\
   \end{array}
 $$
 and the subtree induced by $U=\{1,2,3,4,6,8\}$ is fully leafed because it has $6$ vertices, $4$ of them are leaves, and because $L_G(6) = 4$.
\end{example}

\begin{figure}[h]
 \begin{center}
  \includegraphics[scale=1]{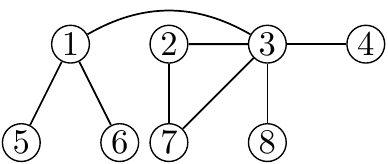}
 \end{center}
 \caption{A graph with vertex set $V=\{1,\ldots,8\}$.}\label{fig:leaf_function}
\end{figure}

\begin{remark}
For any simple graph $G$, we have $L_G(0) = 0$ because the empty tree has no leaf, and $L_G(1) = 0$, since a single vertex is not a leaf. Finally, we always have $L_G(2)=2$ in any graph $G$ with at least one edge.
\end{remark}

The following observations are immediate.
\begin{proposition}
Let $G$ be a connected graph with $n\ge 3$ vertices. If $G$ is non-isomorphic to $K_n$, the complete graph on $n$ vertices, then $L_G(3) = 2$.
\end{proposition}
\begin{proposition} \label{prop:non-decreasing}
For any simple graph $G$ with at least 3 vertices,
the sequence $(L_G(i))_{i=0,1,\ldots,\size{G}}$ is non-decreasing  if and only if $G$ is a tree.
\end{proposition}
\begin{proof}
If $G$ is a tree, then  $L_G(i)$ cannot be decreasing because if a  subtree $T_1$  of $G$ contains a subtree $T_2$ then $\leaf{T_1}\geq \leaf{T_2}$. If $G$ is not a tree, then either $G$ contains a cycle or $G$ is not connected.  In both cases, $G$ has no subtree with $\size{G}$ vertices. Therefore $L_G(\size{G})=-\infty$ and $L_G(2)=2$ which implies that there  exists a decreasing step in the sequence $L_G(i)$.
\end{proof}

We now describe the complexity of solving the problem \ISil{}.

\begin{theorem}
The problem \ISil{} of determining whether there exists an induced subtree with $i$ vertices and $\ell$ leaves in a given graph is NP-complete.
\end{theorem}

\begin{proof}
	It is clear that \ISil{} is in the class NP.
	To show that it is NP-complete, we reduce it to the well-known NP-complete problem \textsc{Independent Set} (\IndepSet{})~\cite{garey}: \textit{Given a graph $G$ and a positive integer $k$, does there exist an independent set of size $k$ in $G$, i.e. a subset of $k$ vertices that are not pairwise adjacent?}
	Note that an instance of \ISil{} is represented by the tuple $(G,i,\ell)$ where $G$ is a graph, $i$ the vertex parameter and $\ell$ the leaf parameter.
	We represent an instance of \IndepSet{} by the tuple $(G,k)$ where $G$ is a graph and $k$ is an integer.

	Consider the map $f$ that associates to an instance $(G,k)$ of
        \IndepSet{} with $G=(V,E)$, the instance $(H,k+1,k)$ of \ISil{} such that the graph $H$ is obtained as $G$ with an additional universal vertex $u$, that is linked to each vertex of $G$.
	Clearly, the map $f$ is computable in polynomial time
	as the graph obtained has $|V|+1$ vertices and $|E|+|V|$ edges.


	If $(G,k)$ is a positive instance of \IndepSet{}, i.e.\ an instance for which the answer is yes, then $f(G,k)=(H,k+1,k)$ is a positive instance of \ISil{}. Indeed, assume that the graph $G$ has an independent set of size $k$. Then these vertices together with the universal vertex $u$ is an induced subtree of $H$ with $k+1$ vertices and $k$ leaves.
	Conversely, if the instance $f(G,k)=(H,k+1,k)$ is a positive instance of \ISil{}, then $(G,k)$ is a positive instance of \IndepSet{}.
	Indeed, assume that $H$ contains an induced subtree $T$ with $k+1$ vertices and $k$ leaves. Observe that the universal vertex $u$ cannot be a leaf unless $k+1\le 2$.
	Consider first that $k>1$. Then the subtree leaves form an independent set of size $k$ of $H$ and of $G$.
	Consider now that $k=1$. As $H$ contains an induced subtree $T$ with $k+1=2$ vertices, $G$ has at least one vertex, which is an independent set of size $1$.
	For $k=0$, $(G,0)$ is clearly a positive instance of \IndepSet{}.

	Therefore, \IndepSet{} $\le$ \ISil{} and \ISil{} is NP-complete.
\end{proof}

From this reduction, we obtain insights on the parameterized complexity of  \ISil{} problem.
A problem, which is parameterized by $k_1,\ldots,k_j$, is said to be \emph{fixed parameter tractable} if it can be solved in time $\bigo(f(k_1,\ldots,k_j)n^c)$ where $n$ is the size of the input, $c$ is a constant independent from the parameters $k_1,\ldots,k_j$ and $f$ is a function of $k_1,\ldots,k_j$.
The class $\mathrm{FPT}$ contains all parameterized problems that are \emph{fixed parameter tractable}.
Similarly to the conventional complexity theory, Downey and Fellows introduced a hierarchy of complexity classes to describe the complexity of parameterized problems \cite{Downey-Fellows-1999}:  $\mathrm{FPT}\subseteq \mathrm{W[1]}\subseteq \mathrm{W[2]}\subseteq \ldots$
Since \IndepSet{} is  $\mathrm{W[1]}$-complete \cite{Downey-Fellows-1995b}, it follows that \ISil{} is probably fixed parameter intractable.
Note that when we replace the induced condition with spanning, the problem becomes fixed parameter tractable \cite{Bodlaender-1989,Downey-Fellows-1995}.

\begin{corollary}
	If $\mathrm{FTP}\neq\mathrm{W[1]}$, then \ISil{} $\not\in\mathrm{FPT}$.
\end{corollary}

We end this section by computing the function $L_G(i)$ for well known families of graphs. First, we consider classical families of finite graphs.
Proofs are omitted as they are straightforward.

\paragraph{Complete graphs $K_n$.}
For the complete graph with $n$ vertices,
$$L_{K_n}(i) = \begin{cases}
  0,       & \mbox{if $i = 0, 1$;} \\
  2,       & \mbox{if $i = 2$;} \\
  -\infty, & \mbox{if $3 \leq i \leq n$;}
\end{cases}$$
since any induced subgraph of $K_n$ with more than two vertices contains a cycle.

\paragraph{Cycles $\mathcal{C}_n$.}
For the cyclic graph $\mathcal{C}_n$ with $n$ vertices, we have
$$L_{\mathcal{C}_n}(i) = \begin{cases}
  0,       & \mbox{if $i = 0, 1$;} \\
  2,       & \mbox{if $2 \leq i < n$;} \\
  -\infty, & \mbox{if $i = n$.}
\end{cases}$$
\paragraph{Wheels $W_n$.} For the wheel $W_n$ with $n+1$ vertices,
$$L_{W_n}(i) = \begin{cases}
  0,       & \mbox{if $i = 0, 1$;} \\
  2,       & \mbox{if $i = 2$;} \\
  i-1,     & \mbox{if $3 \leq i \leq \lfloor\frac{n}{2}\rfloor+1$;} \\
  2,       & \mbox{if $\lfloor\frac{n}{2}\rfloor+2 \leq i \leq n - 1$;} \\
  -\infty, & \mbox{if $n \le i \le n+1$.}
\end{cases}$$
 \paragraph{Complete bipartite graphs $K_{p,q}$.}
For the complete bipartite graph $K_{p,q}$ with $p+q$ vertices,
$$L_{K_{p,q}}(i) = \begin{cases}
  0,       & \mbox{if $i = 0, 1$;} \\
  2,       & \mbox{if $i = 2$;} \\
  i-1,     & \mbox{if $3 \le i \le \max(p,q)+1$;} \\
  -\infty, & \mbox{if $\max(p,q)+2 \leq i \leq p+q$.}
\end{cases}$$

\paragraph{Hypercubes $Q_d$.} For the hypercube graph $Q_d$ with $2^d$ vertices, the computation of $L_{Q_d}$ is more intricate. Using the branch and bound algorithm described in Section~\ref{sec:algo} and implemented in \cite{GitHub}, we were  able to compute the values of the function $L_{Q_d}$ for $d \leq 6$ (see Table~\ref{tab:qn}).

\begin{table}[htbp]
\centering
\footnotesize
\begin{tabular}{c|ccccccccccccccccccccc|}
$n$&0&1&2&3&4&5&6&7&8&9&10&11&12&13&14&15&16&17\\
\hline
$L_{Q_2}(n)$&0&0&2&2&*&&&&&&&&&&&&&\\
$L_{Q_3}(n)$&0&0&2&2&3&2&*&*&*&&&&&&&&&\\
$L_{Q_4}(n)$&0&0&2&2&3&4&3&4&3&4&*&*&*&*&*&*&*&\\
$L_{Q_5}(n)$&0&0&2&2&3&4&5&4&5&6&6&6&7&7&7&8&8&8\\
$L_{Q_6}(n)$&0&0&2&2&3&4&5&6&5&6&7&8&8&9&9&10&10&11\\[5mm]
$n$&18&19&20&21&22&23&24&25&26&27&28&29&30&31&32&33&34&$\ldots$\\
\hline
$L_{Q_5}(n)$&*&*&*&*&*&*&*&*&*&*&*&*&*&*&*\\
$L_{Q_6}(n)$&11&12&12&13&13&14&14&15&15&16&16&17&17&18&18&18&*&$\ldots$\\
\end{tabular}
\caption{The leaf function $L_{Q_d}(i)$ for $2 \leq d \leq 6$. The symbol $*$ is used instead of $-\infty$ to gain some space.}
\label{tab:qn}
\end{table}%

\paragraph{Infinite planar lattices.}
Blondin Mass\'e \textit{et al.}\ have computed the map $L_{\Squ}(i)$, where $\Squ = (\Z^2,A_4)$ is the regular square lattice with respect to the $4$-adjacency relation $A_4$~\cite{Blondin-FPSAC}:
\begin{equation*}
  L_{\Squ}(i) = \begin{cases}
    0,                   & \mbox{if $i = 0,1$;} \\
    2,                   & \mbox{if $i = 2$;} \\
    i - 1,               & \mbox{if $i = 3,4,5$;} \\
    L_{\Squ}(i - 4) + 2, & \mbox{if $i \geq 6$.}
  \end{cases}\quad
\end{equation*}
A similar argument leads to the computation of $L_{\Hex}(i)$ and $L_{\Tri}(i)$ for the hexagonal and the triangular lattices:
\begin{equation*}
  L_{\Hex}(i) = \begin{cases}
    0,                   & \mbox{if $i = 0,1$;} \\
    2,                   & \mbox{if $i = 2, 3$;} \\
    L_{\Hex}(i - 2) + 1, & \mbox{if $i \geq 4$;} \\
  \end{cases}
  \qquad
  \text{and}
  \qquad
  L_{\Tri}(i) = \begin{cases}
    0,                   & \mbox{if $i = 0,1$;} \\
    2,                   & \mbox{if $i = 2, 3$;} \\
    L_{\Tri}(i - 2) + 1, & \mbox{if $i \geq 4$.} \\
  \end{cases}
\end{equation*}
In the three previous cases, the leaf functions verify linear recurrences. It is therefore easy to deduce that their asymptotic growth is $i/2$.  Notice that the functions $L_{\Hex}$ and $L_{\Tri}$ are identical.

\paragraph{The infinite cubic lattice.}
The authors of~\cite{Blondin-FPSAC} also gave the maximal number of leaves $L_{\Cub}(i)$  in induced subgraphs with $i$ vertices for the regular cubic lattice with respect to the $6$-adjacency relation. This leaf function also satisfies a linear recurrence with asymptotic growth $28i/41$ which is slightly larger than for the two-dimensional lattices.

\begin{equation*}\label{defl3}
  L_{\Cub}(i) = \begin{cases}
    0,                     & \mbox{if $i = 0,1$;} \\
  	f(i)+1,		             & \mbox{if $i=6,7,13,19,25$;}\\
    f(i),                  & \mbox{if $2\le i \le 40$ and $i\neq 6,7,13,19,25$;} \\
    f(i - 41)+28,          & \mbox{if $41\le i \le 84$;} \\
    L_{\Cub}(i - 41) + 28, & \mbox{if $i \ge 85$;}
  \end{cases}
\end{equation*}
where $f$ is the function defined by
\begin{equation*}
  f(i) = \begin{cases}
    \lfloor(2i+2)/3\rfloor,	& \mbox{if $0\le i\le 11$;} \\
    \lfloor(2i+3)/3\rfloor,  & \mbox{if $12\le i \le 27$;} \\
    \lfloor(2i+4)/3\rfloor, & \mbox{if $28\le i \le 40$.}
  \end{cases}
\end{equation*}
\noindent

\section{Computing the leaf function of a graph}\label{sec:algo}

We now describe a branch and bound algorithm that computes the leaf function $L_G(i)$ for an  arbitrary simple graph $G$. We propose an algorithm based on a data structure that we call an \emph{induced subtree configuration}.

\begin{definition}

Let $G = (V,E)$ be a simple graph and $\Gamma = \{\green,\yellow,\red,\blue\}$ be a set of colors with coloring functions $c : V \rightarrow \Gamma$. An \emph{induced subtree configuration} of $G$ is an ordered pair $C = (c,H)$, where $c$ is a coloring and $H$ is a stack of colorings called the \emph{history} of $C$.

All colorings $c : V \rightarrow \Gamma$ must satisfy the following conditions for any $u,v \in V$:
\begin{enumerate}
    \item[(i)] The subgraph induced by $c^{-1}(\green)$ is a tree;
    \item[(ii)] If $c(u) = \green$ and $\{u,v\} \in E$, then $c(v) \in \{\green, \yellow, \red\}$;
    \item[(iii)] If $c(u) = \yellow$, then $|c^{-1}(\green) \cap N(u)| = 1$, where $N(u)$ denotes the set of neighbors of $u$.
\end{enumerate}
The \emph{initial induced subtree configuration} of a graph $G$ is the pair $(c_\blue, H)$ where $c_\blue(v)=\blue$ for all $v\in G$ and $H$ is the empty stack. When the context is clear, $C$ is simply called a \emph{configuration}.
\end{definition}

Roughly speaking, a configuration is an induced subtree enriched with information that allows one to generate other induced subtrees either by extension, by exclusion or by backtracking.
The colors assigned to the vertices can be interpreted as follow. The \green{} vertices are the confirmed vertices to be included in a subtree. Since  each \yellow{}  vertex is connected to exactly one \green{} vertex, any  \yellow{}  vertex can be safely added to the green subtree to create a new induced subtree.
The \red{} vertices are those that are excluded from any possible tree extension. A \red{} vertex is excluded  by calling the operation $\Call{ExcludeVertex}{}$ which is defined below.
 The exclusion of a \red{} vertex is done either because it is adjacent to more than one \green{} vertex and its addition would create a cycle or because it is explicitly excluded for generation purposes.  Finally, the \blue{} vertices are available vertices that have not yet been considered  and that could be considered later.
For reasons that are explained in the next paragraphs, it is convenient to save in the stack $H$ the colorations from which $C$ was obtained.

Figure~\ref{fig:graph_border}(a) illustrates an induced subtree configuration. The  \green{} vertices and edges are outlining the induced subtree. The \yellow{} vertices and edges are showing the possible extensions of the \green{} tree. The vertices $14$ and $15$ are  \red{} because each one is connected to two \green{} vertices. Although the vertex $9$ is colored in \red{}, it would have been possible to color it in \yellow{} because  it is connected to exactly one \green{} vertex. Similarly, vertices $12$, $13$ and $16$ could be colored either in \blue{} or \red{} since they are not adjacent to the tree.

\begin{figure}[htb]
    \begin{center}
        \includegraphics[scale=0.8]{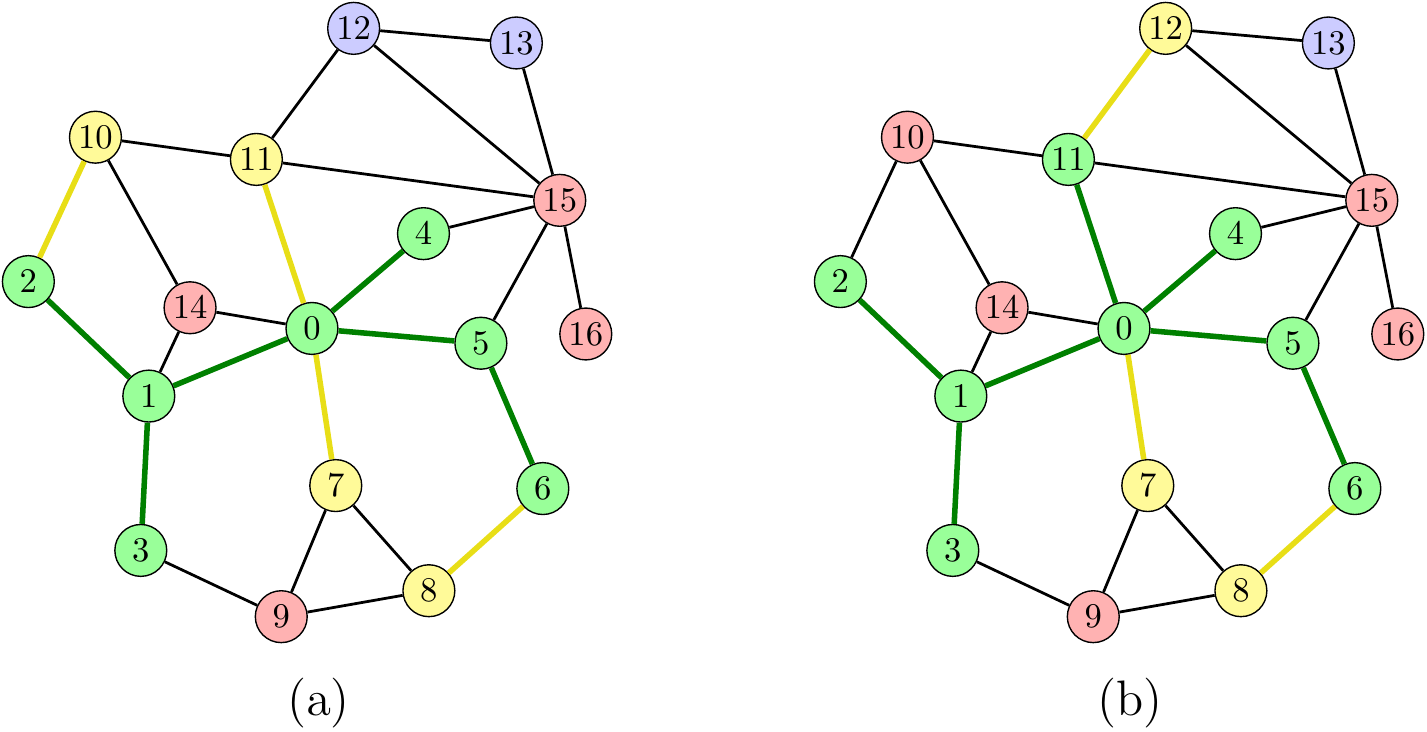}
    \end{center}
    \caption{Induced subtree configurations. The green edges outline the green subtree and the yellow edges outline the possible extensions. (a) A configuration $C$. (b) The configuration $C.\textsc{AddToSubtree}(11)$.}\label{fig:graph_border}
  \end{figure}
Let $C = (c, H)$ be a configuration of some simple graph $G = (V,E)$, with coloring $c$ and  stack $H$.
We consider the following operations  on $C$:

\begin{itemize}
  \item $C.\Call{VertexToAdd}{\null}$ is a non deterministic function that returns any non \green{} vertex in $G$ that can be safely colored in \green{}.
  If no such vertex exists, it returns \emph{none}. Note that the color of the returned vertex is always $\yellow$, except when $c^{-1}(\green)=\emptyset$, where the color is $\blue$.
  \item $C.\Call{AddToSubtree}{v}$ first pushes a copy of $c$ on top of $H$, sets the color of $v$ to \green{} and updates the colors of the neighborhood of $v$ accordingly.
  Notice that this operation is applied only to a vertex $v$ that can be safely colored in green.
  \item $C.\Call{ExcludeVertex}{v}$ first pushes a copy of $c$ on top of $H$ and then sets the color of $v$ to \red{}. This operation is applied only on a vertex $v$ such that $c(v) \in \{\yellow, \blue\}$.
  \item $C.\Call{Undo}{\null}$ retrieves and removes the top of $H$, then stores it into $c$. In other words, this operation  cancels the last operation applied on $C$, which is either an inclusion or an exclusion.
\end{itemize}
To illustrate these operations, let $C$ be the configuration in Figure~\ref{fig:graph_border}(a). Then $C.\textsc{VertexToAdd}{}()$ could return
 one of the  \yellow{} vertices $7$, $8$, $10$ or $11$. Let $C'$ be the configuration obtained from $C$ after calling $C.\Call{AddToSubtree}{11}$. Then we  have to update the colors of vertices $10$, $11$ and $12$ by setting $c(11) \gets \green$, $c(10) \gets \red$ and $c(12) \gets \yellow$,
 as illustrated in Figure~\ref{fig:graph_border}(b).
For any configuration $C$, we call $C'$ an \emph{extension} of $C$ when its coloration is obtained from $C$ without backtracking, i.e.  by using only $\Call{AddToSubtree}{v}$ and $\Call{ExcludeVertex}{v}$.

For optimization purposes, it is worth mentioning that it is not necessary to keep a complete copy of the colorations when they are saved in the history $H$. It is sufficient to store the vertex which caused a vertex to become $\red$ together with a stack of the vertices on which the operations was performed.
Keeping this optimization in mind, it is easy to show that the operations $\Call{AddToSubtree}{v}$ and $\Call{Undo}{}()$, in the case where the last operation is an inclusion of a vertex $v$, are done in $\bigo(\deg(v))$ time.
Also, the operations $\Call{ExcludeVertex}{v}$ and $\Call{Undo}{}()$, in the case where the last operation is an exclusion of a vertex $v$, are done in $\bigo(1)$ time.
Hence, we do not need to copy the whole coloring function at each step, but simply update the neighborhood of some vertex.

It is quite straightforward to use configurations for the generation of all induced subtrees of a graph $G$.
Starting with the initial configuration, it is sufficient to recursively build configurations by branching according to whether some vertex $v$ returned by the operation $C.\Call{VertexToAdd}{\null}$ is included or excluded from the current green tree.
Considering this process as a tree of configurations, the operation can be paired with edges of this tree. Therefore, a careful analysis shows that the generation runs in $O(|V|)$ amortized per solution.

While iterating over all possible configurations, if we want to compute the leaf function $L_G$, it is obvious that some configurations should be discarded whenever they cannot extend to interesting configurations.
Therefore, given an induced subtree configuration of $n$ \green{} vertices, we define the function $C.\Call{LeafPotential}{n'}$, for $n \leq n' \leq |V|$, which computes an upper bound on the number of leaves that can be reached by extending the current configuration $C$ to a configuration of $n'$ \green{} vertices.
First,  the  potential is  $-\infty$  for  $n'$ greater  than  the  size of  the
connected  component  $K$  containing  the green  subtree  (when  computing  the
connected component, we treat red vertices as removed from the graph).
Second, in order to  compute this upper  bound for $n'\leq |K|$, we  consider an
optimistic  scenario in
which all  \yellow{} and \blue{}  vertices that are close  enough can
safely  be  colored  in \green{} without creating a cycle,  whatever the order in which they are selected.
Keeping  this  idea  in  mind, we  start  by
partitioning the \emph{available} vertices, which are the \yellow{} and \blue{} vertices together with the
leaves  of the \green{} tree, according to
their distance  from the inner  vertices of the configuration subtree in $K$.
Algorithm~\ref{algo:leaf_potential} computes an upper bound for the number of leaves that can be realized from a configuration of $n$ \green{} vertices extended to a configuration of $n'$ \green{} vertices.

\begin{algorithm}[htb]
    \begin{algorithmic}[1]
      \Function{LeafPotential}{$C$ : configuration, $n'$ : natural}{: natural}
        \State $n\gets $ number of \green{} vertices
        \State $\ell \gets $ number of leaves in the \green{} subtree
        \State $y \gets $ number of \yellow{} vertices adjacent to an inner \green{} subtree vertex
        \If{$n+y\geq n'$}\label{line:complete-start}
            \State $(n, \ell)\gets (n',\ell + (n'-n))$
        \Else
            \State $(n,\ell) \gets (n+y,\ell+y)$
        \EndIf\label{line:complete-end}
        \State $d\gets 1$
        \While{$n < n'$ and there exists an available vertex at distance at most $d$}
            \State Let $v$ be an available vertex of highest degree
            \Comment{The degree does not count \red{} vertices}
            \If{$n+\deg(v)-1\leq n'$}\label{line:available-start}
                \State $(n,\ell) \gets (n+\deg(v)-1, \ell+\deg(v)-2)$\label{line:update}
            \Else
                \State $(n, \ell)\gets (n',\ell+(n'-n)-1)$
            \EndIf\label{line:available-end}
            \State Remove $v$ from available vertices
            \State $d\gets d+1$
        \EndWhile
        \State \Return $\ell$
      \EndFunction
    \end{algorithmic}
    \caption{Computation of the leaf potential for $n'$} \label{algo:leaf_potential}
\end{algorithm}

The first part of Algorithm~\ref{algo:leaf_potential} consists in \emph{completing} the \green{} subtree. More precisely, a configuration $C$ is called \emph{complete} if each \yellow{} vertex is adjacent to a leaf of the \green{} tree. We first verify if $C$ is complete and, when it is not the case, we increase $n$ and $\ell$ as if the \green{} subtree was completed  (Lines \ref{line:complete-start}--\ref{line:complete-end}).
Next, we choose a vertex $v$ among all available vertices within distance $d$. We assume that $v$ is \green{} and  update $n$ and $\ell$ as if all non-\green{} neighbors of $v$
 were leaves added to the current configuration (Lines \ref{line:available-start}--\ref{line:available-end}). This process is repeated until the size of the ``optimistic subtree'' reaches $n'$.
 Note that this process never decreases the values of $n$ and $\ell$. Indeed, in Line \ref{line:update}, the degree of $v$ is always greater than $1$ since $n'$ does not exceed the size of the connected component $K$.

\begin{remark}\label{rem:invariant}
 We note that $(n,\ell)$ in Algorithm~\ref{algo:leaf_potential} always satisfies the loop invariant $n-\ell = |I|+k$ where $I$ is the set of inner vertices of the \green{} subtree of $C$ and $k$ is the number of iterations of the loop. When Algorithm~\ref{algo:leaf_potential} ends after $k'$ iterations of the loop, the output is $n'-(|I|+k')$ as $n=n'$.
\end{remark}

We now prove that Algorithm~\ref{algo:leaf_potential} yields an upper bound on the maximum number of leaves that can be realized. It is worth mentioning that, in order to obtain a nontrivial bound, we restrict the available vertices to those that are within distance $d$ from the inner vertices of the current \green{} subtree,
 and then we increase the value of  $d$ at each iteration.

\begin{proposition}\label{P:upper}
Let $C$ be a configuration of a simple graph $G = (V,E)$ with $n \geq 3$ \green{} vertices and let $n'$ be an integer such that $n \leq n' \leq |V|$. Then any extension of $C$ to a configuration of $n'$ vertices has at most $C.\Call{LeafPotential}{n'}$ leaves, where $C.\Call{LeafPotential}{n'}$ is the operator described in Algorithm~\ref{algo:leaf_potential}.
\end{proposition}
\begin{proof}
Let $\ell$ be the number of leaves of the \green{} subtree represented by $C$, $I$ be the set of inner vertices in the \green{} subtree
 and
$$Y=\{v : v \text{ is a \yellow{} vertex of $C$ at distance 1 from $I$}\}.$$
Let $p' = C.\Call{LeafPotential}{n'}$.
If $n'-n\le |Y|$, then $p'=\ell +n'-n$ and it is clear that adding $n'-n$ vertices cannot add more than $n'-n$ leaves.

Otherwise, we proceed by contradiction
by assuming that $C$ can be extended to a configuration $C'$ with $n'$ \green{} vertices and $\ell'$ leaves, with $\ell' > p'$. Let $v_1, v_2, ..., v_k$  be the sequence of vertices that became, in that order, inner vertices in the successive extensions of $C$ to reach $C'$.
Then we have $\ell'=n'-(|I|+k)$.
Let $v'_1, v'_2, \ldots, v'_{k'}$ be the vertices chosen by the procedure $C.\Call{LeafPotential}{n'}$.
It follows from Remark~\ref{rem:invariant} that $p'=n'-(|I|+k')$.
As we assumed that $\ell'>p'$, we obtain $k < k'$.

Without loss of generality, we assume that if $v_i$ and $v_j$ are at the same distance from  $I$ and $\deg(v_j)\leq \deg(v_i)$ then $i \leq j$ (otherwise, we simply swap any pair of vertices $v_i$ and $v_j$ that do not satisfy this condition).
Moreover, we know that $v_i$ is at most at distance $i$ from $I$.
Hence,
$$\deg(v_1)\leq \deg(v'_1),\; \deg(v_2)\leq \deg(v'_2),\; ...,\; \deg(v_k)\leq \deg(v'_k).$$
Therefore, for each new inner vertex $v_i$, only its neighbors can be included without adding an inner vertex. Similarly, including $v_i$ as an inner vertex implies that at most $\deg(v_i)-2$ leaves are gained. Taking into account the potential leaves found in $Y$, we conclude that
$$\ell' \leq \ell + |Y|+\sum_{i=1}^k (\deg(v_i)-2)  \leq \ell + |Y|+\sum_{i=1}^{k'-1} (\deg(v'_i)-2) \le p'$$
which is a contradiction, showing that the configuration $C'$ cannot exist.
\end{proof}

It follows from Proposition~\ref{P:upper} that a configuration $C$ of $n$ \green{} vertices and $r$ \red{} vertices cannot be extended to a configuration
whose subtree has more leaves than prescribed by the best values found for $L$ so far when
\begin{equation} \label{eq:promissing_condition}
    C.\Call{LeafPotential}{n'} \leq L(n') \text{ for all }n\leq n' \leq |K|.
\end{equation}

We conclude this section by presenting Algorithm~\ref{algo:l_computation}, which computes the function $L$ for an arbitrary simple graph $G$. The idea guiding this algorithm simply consists in generating all possible configurations,
 discarding those that cannot be extended to fully leafed induced subtrees.

\begin{algorithm}[htb]
    \begin{algorithmic}[1]
        \Function{LeafFunction}{$G$: graph}{: list}
            \Function{ExploreConfiguration}{\null}{}
                \State $u \gets C.\Call{VertexToAdd}{\null}$
                    \If{$u = \emph{none}$}
                        \State $i \gets$ the number of \green{} vertices in $C$
                        \State $\ell \gets$ the number of leaves in $C$
                        \State $L[i] \gets \max(L[i], \ell)$
                    \ElsIf{Inequation \eqref{eq:promissing_condition} is not satisfied}
                        \State $C.\Call{AddToSubtree}{u}$
                        \State $\Call{ExploreConfiguration}{\null}$
                        \State $C.\Call{Undo}{\null}$
                        \State $C.\Call{ExcludeVertex}{u}$
                        \State $\Call{ExploreConfiguration}{\null}$
                        \State $C.\Call{Undo}{\null}$
                    \EndIf
            \EndFunction
            \State Let $C$ be the initial configuration of $G$
            \State $L[0] \gets 0$
            \State $L[i] \gets -\infty$ for $i = 1,2,\ldots,\size{G}$
            \State $\Call{ExploreConfiguration}{\null}$
            \State \Return $L$
        \EndFunction
    \end{algorithmic}
    \caption{Leaf function computation} \label{algo:l_computation}
\end{algorithm}

Based on Proposition~\ref{P:upper} and the previous discussion, the following result is immediate.
\begin{theorem}
Let $G$ be any simple graph. Then Algorithm~\ref{algo:l_computation} returns the leaf function $L_G$ of $G$.
\end{theorem}

\begin{figure}[htb]
 \begin{center}
 \begin{tabular}{ccc}
     \includegraphics[scale=0.9]{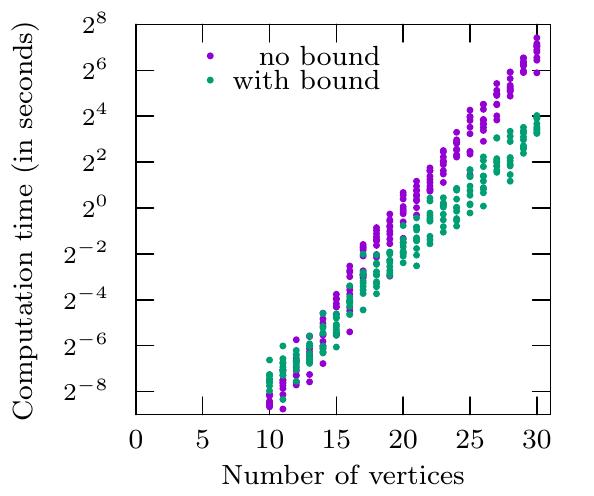} & &
     \includegraphics[scale=0.9]{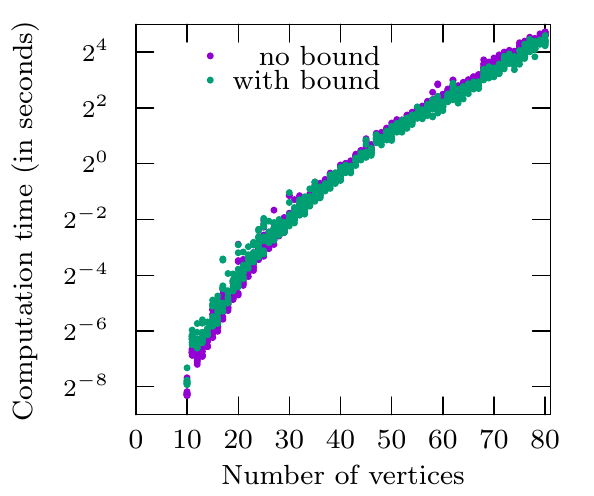}\\
  (a) & & (b) \\
  \includegraphics[scale=0.9]{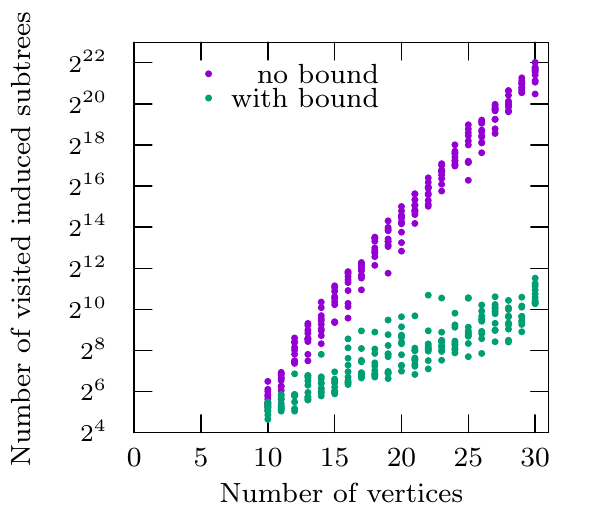} & &
  \includegraphics[scale=0.9]{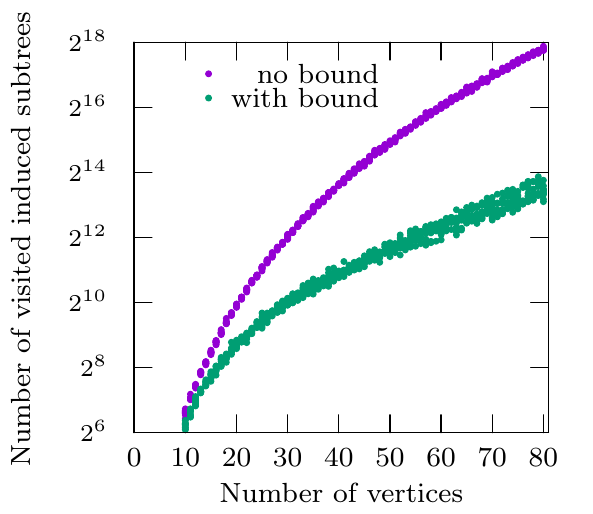}\\
  (c) & & (d) \\
 \end{tabular}
 \caption{
 The running time of Algorithm \ref{algo:l_computation} on 10 randomly generated graphs with density $0.2$ (a) and density $0.8$ (b), with or without using the leaf potential bound. The corresponding number of induced subtrees for density $0.2$ (c) and density $0.8$ (d) that are visited during the execution.
 }\label{fig:density}
 \end{center}
\end{figure}

Empirically, we observed the following elements.
First, it seems that the overall time performance is significantly better on dense graphs.
More precisely, for a fixed number of vertices, the computation of the leaf function is faster on a dense graph than on a sparse one (see Figure~\ref{fig:density}(a-b)).
This is not surprising, since if one takes a vertices subset of a dense graph, the probability that these vertices induce at least one cycle is high.
Therefore, the number of visited induced subtrees is smaller.
For example, experimental data show that the number of visited subtrees in a graph with $30$ vertices and expected density $0.1$ is still around ten times greater than the number of visited subtrees in a graph with $80$ vertices and expected density $0.9$.
Figure~\ref{fig:subtree-by-prob} illustrates the number of induced subtrees according to density in random graphs.

\begin{figure}[htb]
	\centering
        \includegraphics[scale=1]{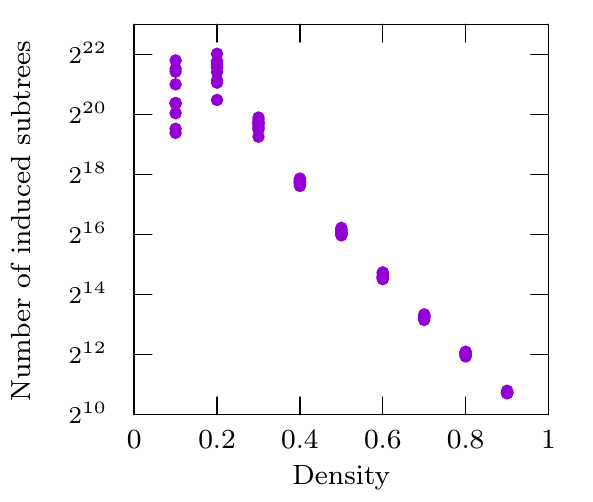}
	\caption{Number of induced subtrees in graphs with 30 vertices, randomly generated according to density.}\label{fig:subtree-by-prob}
\end{figure}

Moreover, the leaf potential bound always reduces the number of visited subtrees regardless of the density.
However, the difference is more pronounced on lower density graphs (see Figure~\ref{fig:density}(c-d)).
This also seems easily explainable: It is expected that, as the density decreases, the number of layers in the vertices partition increases and the degrees of the vertices diminish.
Hence, when we use the leaf potential as a bounding strategy, the computation time gain is more significant on sparse graphs.

Hence, for lower density graphs, the leaf potential improves the algorithm and the overall performance of the algorithm.
For higher density, no significant difference in time performance with or without the usage of the bound is observed.
Finally, from an empirical point of view,
 the number of  visited induced subtrees seems to indicate an overall complexity of the algorithm in $\bigo(\alpha^n)$ with $\alpha<2$.
Unfortunately, we were unable to prove such an upper bound.


\section{Fully leafed induced subtrees of trees}\label{sec:tree}

It turns out that the \MLIS{} problem can be solved in polynomial time when it is restricted to the class of trees.
Observe that since all subtrees of trees are induced subgraphs, we could omit the ``induced'' adjective, but we choose to keep the expression for sake of uniformity.

A naive strategy consists in successively deleting suitable leaves to obtain a sequence of fully leafed  subtrees embedded in each other.  Such a strategy is not viable. Indeed, consider the tree $T$ represented in Figure~\ref{fig:tree2}. We have $L_T(9)=6$ and $L_T(7)=5$ and there is exactly one fully leafed induced subtree of $T$ with respectively 7 and 9 vertices. But the smallest of these two subtrees (in blue) is not a subgraph of the largest one (in red).
\begin{figure}[htbp]
\begin{center}
\includegraphics[scale=0.8]{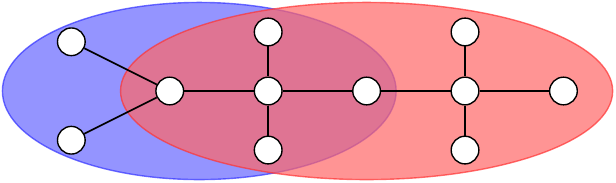}
\end{center}
\caption{A tree with its unique fully leafed induced subtrees with $7$ (respectively $9$) vertices in the blue (resp. red) area.}
\label{fig:tree2}
\end{figure}

Hereafter, we describe an algorithm with polynomial time complexity based on the dynamic programming paradigm but before, we recall some definitions.
A \emph{rooted tree}  is a couple $\widehat{T}=(T,u)$ where $T=(V,E)$ is a tree and $u\in V$ is a distinguished vertex called the \emph{root} of $\widehat{T}$. Rooted trees have a natural orientation with arcs pointing away from the root.  A \emph{leaf} of a rooted tree is a vertex $v$ with outdegree $\deg^+(v) = 0$.  In particular, if a rooted tree consists in a single vertex, then this vertex is a leaf. The functions $\size{\widehat{T}}$ and $\leaf{\widehat{T}}$ are defined accordingly by
$$\size{\widehat{T}}=\size{T}\text{ and }\leaf{\widehat{T}} = \left|\left\{v \in \widehat{T} : \deg^+(v) = 0\right\}\right|.$$
Similarly, a \emph{rooted forest} $\widehat{F}$ is a collection of rooted trees.
It follows naturally that
$$\leaf{\widehat{F}} = \sum_{\mbox{$\widehat{T} \in \widehat{F}$}} \leaf{\widehat{T}}.$$
The \emph{rooted forest induced by a rooted tree} $\widehat{T}=(T,u)$  is the set of rooted trees obtained by removing from $T$ the root $u$ and its incident edges so that the $k$ vertices adjacent to $u$ become roots of the trees $\widehat{T}_i$.
Let $\widehat{T}$ be any rooted tree with $n$ vertices and $L_{\widehat{T}} : \{0,1,\ldots,n\} \rightarrow \N$ be defined by
\begin{equation*}
L_{\widehat{T}}(i) = \max\left\{\leaf{\widehat{T'}} \:\big\vert\: \widehat{T'} \preceq \widehat{T}\text{ and } \size{\widehat{T'}} = i\right\},
\end{equation*}
where $\preceq$ denotes the relation ``being a rooted subtree with the same root''.
Roughly speaking, $L_{\widehat{T}}(i)$ is the maximum number of leaves that can be realized by some rooted subtree of size $i$ of $\widehat{T}$. This map is naturally extended to rooted forests so that for a rooted forest   $\widehat{F}=\{\widehat{T_1},\ldots,\widehat{T_k}\}$
we set
\begin{equation}\label{EQ:L-forest-def}
L_{\widehat{F}}(i)=\max\left\{\sum_{j=1}^k \leaf{\widehat{T'_j}} \:\big\vert\: \widehat{T'_j}\preceq\widehat{T_j} \text{ and }\sum_{j=1}^k \size{\widehat{T'_j}}=i\right\}.
\end{equation}
Let $C(i,k)$ be the set of all weak compositions $\La=(\La_1,\ldots ,\La_k)$ of $i$ in $k$ nonnegative parts. Then Equation~\eqref{EQ:L-forest-def} is equivalent to
\begin{equation}\label{EQ:L-forest}
L_{\widehat{F}}(i) = \max \left\{\sum_{j=1}^k L_{\widehat{T_j}}(\lambda_j)~:~{\lambda \in C(i,k)}\right\}.
\end{equation}
Assuming that $L_{\widehat{T_j}}$ is known for $j = 1,2,\ldots,k$, a naive computation of $L_{\widehat{F}}$ using Equation~\eqref{EQ:L-forest} is not done in polynomial time, since
$$|C(i,k)| = \binom{i + k - 1}{i}.$$
Nevertheless, the next lemma shows that $L_{\widehat{F}}$ can be computed in polynomial time.
\begin{lemma}\label{L:forest-tree}
Let $k\ge1$ be an integer and $\widehat{F}=\{\widehat{T_1},\ldots,\widehat{T_k}\}$ be a rooted forest with $n$ vertices.
Then, for $i\in\{0,\ldots,n\}$,
\begin{equation}\label{EQ:L-forest-efficient}
L_{\widehat{F}}(i) =
    \begin{cases}
      L_{\widehat{T_1}}(i), & \text{if }k=1;\\
      \max\{L_{\widehat{T_1}}(j) + L_{\widehat{F'}}(i - j) : \max\{0,i-\size{\widehat{F'}}\} \leq j \leq \min\{i,\size{\widehat{T_1}}\}\}, &\text{if }k\ge2;\\
    \end{cases}
\end{equation}
where $\widehat{F'}=\{\widehat{T_2},\ldots,\widehat{T_k}\}$. Therefore, if $L_{\widehat{T_j}}$ is known for $j = 1,2,\ldots,k$, then $L_{\widehat{F}}$ can be computed in $\bigo(kn^2)$ time.
\end{lemma}
\begin{proof}
The first part follows from Equation~\eqref{EQ:L-forest} and the fact that, for $k \geq 2$, we have
$$C(i,k) = \{(j, \lambda_1, \lambda_2, \ldots, \lambda_{k-1}) : 0 \leq j \leq i,
(\lambda_1, \lambda_2, \ldots, \lambda_{k-1}) \in C(i-j,k-1)\}.$$

For the time complexity, one notices that for a given $i$, the recursive step of Equation~\eqref{EQ:L-forest-efficient} is applied $k-1$ times, where each step is done in $\bigo(n)$. Since $L_{\widehat{F}}(i)$ is computed for $i = 1,2,\ldots,n$, the total time complexity is $\bigo(kn^2)$.
\end{proof}

Finally, we describe how  $L_{\widehat{T}}$ is computed from the children of its root.
\begin{lemma}\label{L:tree-forest}
Let $\widehat{T}$ be some rooted tree with root $u$. Let $\widehat{F}$ be the rooted forest induced by the children of $u$. Then
\begin{equation*}
L_{\widehat{T}}(i) = \begin{cases}
  i,                      & \mbox{if $i = 0,1$;} \\
  L_{\widehat{F}}(i - 1), & \mbox{if $2\le i\le \size{\widehat{T}}$.}
\end{cases}
\end{equation*}
\end{lemma}

\begin{proof}
The cases $i = 0,1$ are immediate. Assume that $i \geq 2$. Since any rooted subtree of $\widehat{T}$ must, in particular, include the root $u$ and since $u$ is not a leaf, all the leaves are in $\widehat{F}$ and the result follows.
\end{proof}

Combining Lemmas~\ref{L:forest-tree} and~\ref{L:tree-forest}, we obtain the following result.
\begin{theorem}\label{T:trees}
Let $T = (V,E)$ be an unrooted tree with $n\ge2$ vertices. Then $L_T$ can be computed in $\bigo(n^3\Delta)$ time and $\bigo(n^2)$ space where $\Delta$ denotes the maximal degree of a vertex in $T$.
\end{theorem}

\begin{proof}
Removing any edge $\{u,v\}\in E$  from $T$ gives two subtrees of $T$ and identifying $u$ and $v$ as the roots of these two subtrees allows to recover the edge $\{u,v\}$ and the tree $T$ from them. Therefore we consider
 the  two rooted subtrees  $\widehat{T}(v \rightarrow u)$ rooted in $u$ and   $\widehat{T}(u \rightarrow v)$ rooted in $v$.
Using Lemmas~\ref{L:forest-tree} and~\ref{L:tree-forest}, we compute the values of $L_{\widehat{T}(u \rightarrow v)}$ and  $L_{\widehat{T}(v \rightarrow u)}$  for each edge $\{u,v\}$, and we store the results obtained recursively to avoid duplication in the computation. The overall time complexity is
$$\sum_{\{u,v\} \in E} \left(\bigo(\deg(u)n^2) + \bigo(\deg(v)n^2)\right) = \bigo(n^3\Delta),$$
by Lemma~\ref{L:forest-tree} and the fact that $|E| = n-1$.

Next, let the function $L_{\{u,v\}} : \{0,1,2,\ldots,n\}\rightarrow \mathbb{N}$ be defined by
 $$ L_{\{u,v\}}(i) = \max_{j \in J}\left\{ L_{\widehat{T}(u \rightarrow v)}(j) + L_{\widehat{T}(v \rightarrow u)}(i-j)\right\}$$
where $J=\left\{i\in \mathbb{N}~:~\max\{1,i-\size{\widehat{T}(v\rightarrow u)}\}\leq i \leq \min\{ i - 1,\size{\widehat{T}(u\rightarrow v)}\}\right\}$.
In other words, $L_{\{u,v\}}(i)$ is the maximum number of leaves that can be realized by all subtrees of $T$ with $i$ vertices and containing the edge $\{u,v\}$. Clearly, $L_{\{u,v\}}$ is computed in time $\Theta(n)$ when the functions $L_{\widehat{T}(u \rightarrow v)}$ and $L_{\widehat{T}(v \rightarrow u)}$ have been  computed. Hence, since any optimal subtree with $i \geq 2$ vertices has at least one edge, the optimal value $L_T(i)$
must be stored in at least one edge so that
$$L_T(i) = \begin{cases}
    0, & \mbox{if $i = 0,1$;} \\
    \max\left\{L_{\{u,v\}}(i)~:~\{u,v\} \in E\right\}, & \mbox{if $2\leq i \leq \size{T}$.}
\end{cases}$$
is computed in $\Theta(n)$ time as well. The global time complexity is therefore $\bigo(n^3\Delta)$, as claimed. Finally, the space complexity of $\bigo(n^2)$ follows from the fact that each of the $n - 1$ edges stores information of size $\bigo(n)$.
\end{proof}

\begin{figure}[htb]
  \centering
  \includegraphics[scale=0.8]{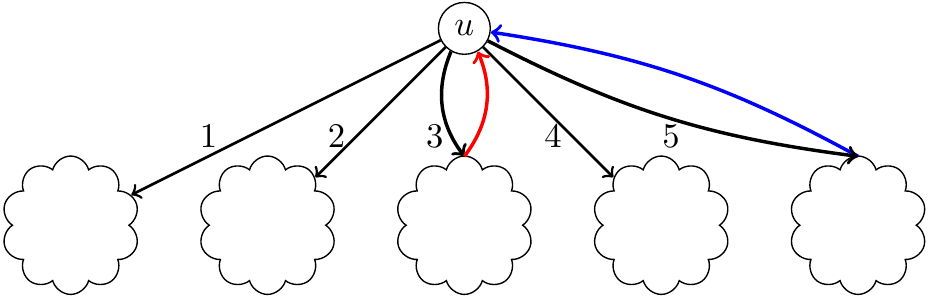}
  \caption{In the computation of the leaf function, the merger of the oriented forest induced by the blue arc can not be used to compute the merger of the oriented forest induced by the red arc.
  }\label{F:delta}
\end{figure}

\begin{remark}
At first sight, it seems that a more careful analysis could lead to a $\bigo(n^3)$ time complexity in Theorem~\ref{T:trees}. However, we have not been able to get rid of the $\Delta$ factor.
Consider a tree $\widehat{T}$ rooted in $u$ and the forest $\widehat{F}$ induced by $\widehat{T}$.
By Lemma~\ref{L:forest-tree}, the computation of $L_{\widehat{F}}$ requires $\deg(u) - 1$ ``merging steps'', i.e., computations using the recursive part of the formula.
Since, in the graph, each arc incident to $u$ induces a different rooted tree and an associated rooted forest,
it does not seem possible to reuse the merger of one forest in the computation of another one.
Therefore, the number of mergers of a given edge can increase up to $\Delta$.
For instance, consider the graph depicted in Figure~\ref{F:delta}. On one hand, the blue arc induces a tree $\widehat{T}_{blue}$ rooted in $u$ including the arcs $1,2,3,4$. The leaf function of the associated forest ${\widehat{F}_{blue}}$ depends on the value stored in the arcs $1,2,3,4$.
On the other hand, the leaf function of the associated forest ${\widehat{F}_{red}}$ for the red arc depends on the value stored in the arcs $1,2,4,5$.
So each arc outgoing from $u$ needs to be merged $\deg(u)-1$ times.
\end{remark}

\begin{figure}[htb]
\begin{center}
\includegraphics[scale=1]{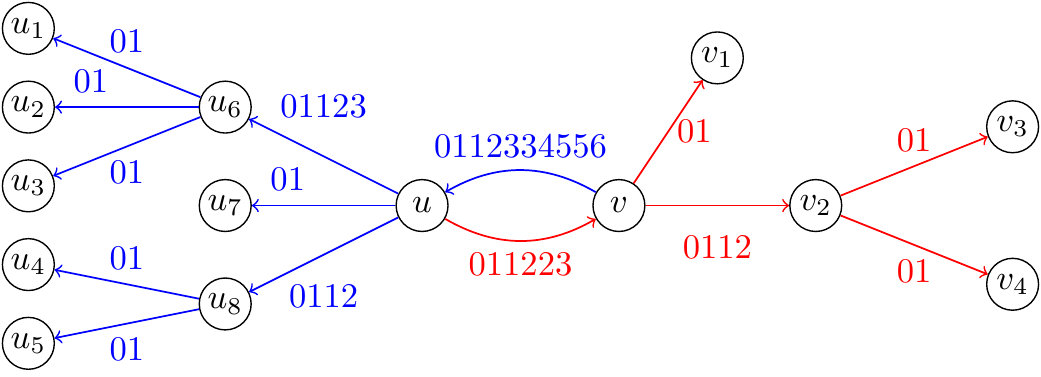}
\end{center}
\caption{Computations of the function $L_{\{u,v\}}$  for the edge $\{u,v\}$.}\label{fig:tree_algo}
\end{figure}

\begin{example}
Consider the tree depicted in Figure~\ref{fig:tree_algo} without orientation and with a single edge $\{u,v\}$.
By Theorem~\ref{T:trees}, the computation of $L_T(i)$ with $i\in\{2,\ldots,14\}$ requires to first compute the function $L_{\{x,y\}}$ for each edge $\{x,y\}\in E$ as shown in Figure~\ref{fig:tree_algo}.
The blue arc $(v,u)$ stores the value of $L_{\widehat{T}(v\to u)}$. As $L_{\widehat{T}(v\to u)}$ is computed recursively on the subtrees rooted in the children, the other blue arcs hold intermediate values necessary for the computation of $L_{\widehat{T}(v\to u)}$. Similarly, the red edges hold the intermediate values of the recursive computation of $L_{\widehat{T}(u\to v)}$.
\end{example}


\section{Perspectives}\label{sec:concl}

There is room for improving and specializing the branch and bound algorithm described in  Section~\ref{sec:algo}.
For  example, we  were able  to speed  up the  computations  for  the hypercube  $Q_6$ by  taking into  account some  symmetries (see  \cite{GitHub}).
In  a  more general  context, we believe that significant improvements could  be obtained  by  exploiting the  complete automorphism group of the graph, particularly in highly symmetric graphs.

We did not discuss the problem of \emph{generating} efficiently the set of all fully leafed induced subtrees.
However, it seems easy to show that, by slightly modifying the branch and bound algorithm and the dynamic programming approach of Section~\ref{sec:tree}, one could generate all optimal induced subtrees with polynomial time delay.

Finally, since the problem \MLIS{} is polynomial for trees, another possible study would be to restrict our attention to special families of graphs.
The classes of 3-colorable graphs, planar graphs and chordal graphs seem promising for finding a polynomial time algorithm, as well as the family of graphs with bounded tree-width.


\bibliographystyle{alpha}
\bibliography{references}
\label{sec:biblio}
\end{document}